\documentclass[11pt]{article}
\usepackage[letterpaper, left=1.25in, right=1.25in, top=1.25in, bottom=1.25in]{geometry}
\usepackage{booktabs} 
\usepackage[ruled,linesnumbered,lined]{algorithm2e} 

\SetAlFnt{\small}
\SetAlCapFnt{\small}
\SetAlCapNameFnt{\small}
\SetAlCapHSkip{0pt}
\IncMargin{-\parindent}

\usepackage{enumitem}

\usepackage{amsmath}
\usepackage{amsthm}
\usepackage{multirow}

\usepackage{graphicx}
\usepackage{amssymb}
\usepackage{amsfonts}
\usepackage{amsmath,amsthm}
\usepackage[footnotesize,bf]{caption}
\usepackage{natbib}
\usepackage{setspace}
\usepackage[dvipsnames,usenames]{color}
\usepackage{pstricks,pstricks-add}
\usepackage{url}
\usepackage[usenames]{color}
\usepackage[breaklinks, colorlinks, citecolor=blue, linkcolor=blue, urlcolor=blue]{hyperref}
\usepackage{lscape}

\newtheorem{theorem}{Theorem}

\newtheorem{definition}{Definition}

\newtheorem{lemma}{Lemma}

\newcommand{\Bcal}{\mathcal{B}}

\newcommand{\Rcal}{\mathcal{R}}

\newcommand{\Fcal}{\mathcal{F}}

\newcommand{\R}{\mathbb{R}}

\newpsobject{grilla}{psgrid}{subgriddiv=1,griddots=10,gridlabels=6pt}

\setcitestyle{acmnumeric}

\begin{document}

\title{A polynomial algorithm for maxmin and minmax envy-free rent division on a soft budget\footnote{Thanks to Ariel Procaccia and seminar participants at the 2019 Economic Design and Algorithms in St.\! Petersburg Workshop and UT Austin for useful comments. The results of this paper circulated previously in a paper that is now solely dedicated to the incentives issues in rent allocation with soft budgets \citep{Velez-Axiv-BudgetIncentives}. All errors are my own.}}
\date{\today}

\author{Rodrigo A. Velez\thanks{
\href{mailto:rvelezca@tamu.edu}{rvelezca@tamu.edu}; \href{https://sites.google.com/site/rodrigoavelezswebpage/home}{https://sites.google.com/site/rodrigoavelezswebpage/home}} \\\small{\textit{Department of Economics, Texas A\&M University, College Station, TX 77843}}}
\maketitle

\begin{abstract}
\begin{singlespace}
The current practice of envy-free rent division, lead by the fair allocation website Spliddit, is based on quasi-linear preferences. These preferences rule out agents' well documented financial constraints. To resolve this issue we consider piece-wise linear budget constrained preferences. These preferences admit differences in agents' marginal disutility of paying rent below and above a given reference, i.e., a soft budget. We construct a polynomial algorithm to calculate a maxmin utility envy-free allocation, and other related solutions, in this domain.
\end{singlespace}
\medskip
\begin{singlespace}

\medskip

\medskip

\end{singlespace}
\end{abstract}

\section{Introduction}

\subsection{Overview of the problem}

The envy-free rent division problem is one of the success stories of computational fair division. It addresses the allocation of rooms and payments of rent among roommates who lease an apartment or a house. The objective is to find a recommendation in which each roommate finds her assignment is at least as good as that of the other roommates.

Research on this problem first concentrated on quasi-linear environments. An early result showed the existence of a polynomial algorithm to calculate an envy-free allocation \citep{Aragones-1995-SCW}.\footnote{All algorithms we refer to and construct are polynomial in the number of agents.} More recently, \citet{Gal-et-al-2016-JACM} argued that there are compelling differences among envy-free allocations, presented evidence that these differences are perceived by agents, identified the \emph{maxmin utility envy-free allocations} as possible candidates to minimize these issues, and constructed a polynomial algorithm to calculate one of these allocations.

Passing quickly from theory to practice, the fair allocation website Spliddit \citep{GP14} implemented \citet{Gal-et-al-2016-JACM}'s approach.  Along the tens of thousands of instances in which the algorithm has been put in practice, its users have also provided a series of requests for its improvement. From this feedback, \citet{Procaccia-Velez-Yu-2018-AAAI} identified the inability of the system to handle budget constraints as its main shortcoming.

Two partial answers have been proposed so far to this issue. First, \citet{Procaccia-Velez-Yu-2018-AAAI} elicit agents' values for rooms and hard budget constraints, and construct a polynomial algorithm to determine whether an envy-free allocation satisfying the budget constraints exists. Whenever possible their algorithm returns a maxmin utility allocation constrained to the envy-free allocations that satisfy budget constraints. When no such allocation exists, so it is necessary that at least one agent pays above her stated budget constraint, \citet{Procaccia-Velez-Yu-2018-AAAI} algorithm has the limitation that it calculates an allocation that is envy-free with respect to the underlying reported quasi-linear preference. This recommendation is not informed by the agents' preferences to go over their budget. The second approach, advanced by \citet{Eshwar-et-al-2018-Arxiv}, addresses this issue. They propose to enlarge the space of admissible preferences to include piece-wise linear preferences as follows: Each agent has a reference rent (a soft budget) that determines a breaking point in the marginal disutility of paying rent (see functional form in Sec.~\ref{Section-model}). We will refer to these as \emph{budget constrained piece-wise linear} preferences. These authors construct a polynomial algorithm to calculate an envy-free allocation for a fixed finite set of possible values for the marginal disutility of paying rent above the agents' budget. The algorithm has the limitation that it loses control of the location of the constructed allocation within the envy-free set. Since some envy-free allocations are intuitively biased \citep{Gal-et-al-2016-JACM}, this approach runs the risk to make such a recommendation.

The objective of this paper is then is to study the computation of particular selections of the envy-free set with budget constrained piece-wise linear preferences. Of particular interest to us are the maxmin utility envy-free allocations.

\subsection{Our contribution}

We construct a family of polynomial algorithms, with the same fixed parameters of\linebreak \citet{Eshwar-et-al-2018-Arxiv}, that compute allocations selected by the maxmin and minmax utility and maxmin and minmax money linear selections of the envy-free set (Theorem~\ref{Th:main-takeout-complexity}). Sec.~\ref{Section-model} presents precise definitions. For the moment, it is relevant to add that the following prominent selections of the envy-free set are particular cases covered by the theorem.
\begin{itemize}
\item Maxmin utility envy-free allocation: maximizes the minimal utility across agents within the envy-free set. This selection is advocated by \citet{Gal-et-al-2016-JACM}.
\item Best envy-free allocation for a given agent: maximizes a given agent's utility within the envy-free set. This selection is minimally manipulable on several partial orders of manipulability \citep{Andersson-Ehlers-Svensson-2014-TE}.
\item Maxmin rent envy-free allocation: maximizes the minimum rent paid by any agent. Calculating such an allocation determines the existence of envy-free allocations in which no agent is compensated to receive a room (the non-negativity restriction is usually avoided in the literature).
\item Minmax rent envy-free allocation: minimizes the maximal rent paid by any agent. This selection implements the tightest uniform rent control that is compatible with no-envy.
\item An allocation that minimizes, among all envy-free allocations, the rent paid by the agent who is assigned a given room.
\item An allocation that maximizes, among all envy-free allocations, the rent paid by the agent who is assigned a given room.
\end{itemize}

\subsection{Significance for CS and Economics}

Our work contributes in two independent levels. First, we address an issue that has been identified from real-world requests in one of the most successful application of computational social choice. As such, our proposals have the potential to be deployed in practice. Second, at a technical level, we bridge two independent branches of the literature on envy-free rent division. On the one hand, we work on the algorithmic framework developed by \citet{Gal-et-al-2016-JACM} and \citet{Eshwar-et-al-2018-Arxiv}. On the other hand, we manage to exploit from a computational perspective the topological properties of the envy-free set that economists have used in unrestricted continuous economies to understand the normative and structural properties of these allocations \citep[c.f.,][]{A-D-G-1991-Eca,Velez-2016-Rent,Velez-2017-Survey}. Showing that this can be done constitutes in itself a case study of the successful interaction between these fields.

\section{Model}\label{Section-model}

A set of $n$ rooms, $A:=\{a,b,\dots\}$, is to be allocated among $n$~agents $N:=\{1,...,n\}$. Each agent is to receive a room and pay an amount of money for it. Agent~$i$ generic allotment is $(r_a,a)\in\R\times A$. When $r_a\geq 0$ we interpret this as the amount of money agent $i$ pays to receive room $a$. We allow for negative payments of rent, i.e., $r_a<0$.\footnote{It is possible to restrict preferences to a cartesian domain that guarantees payments of rent are always nonnegative at each envy-free allocation \citep{Velez-2011-JET}. Our results obviously apply to these subdomains. Most importantly, by considering the unrestricted problem, we are able to construct a polynomial algorithm that determines whether an envy-free allocation in which no agent is compensated exists for a particular preference profile.}

Each agent has preferences on bundles of rooms and payments of rent represented by a utility function of the following form. There are $(v^i_a)_{a\in A}\in\R^A$, $b_i\geq 0$, and $\rho_i\in \R_+$ such that
\[u_i(r_a,a)=v^i_a-r_a-\rho_i\max\{0,r_a-b_i\}.\]
For a finite set  $\{\rho_1,\dots,\rho_k\}\subseteq\R_+$ that we fix throughout, the space of preferences in which the coefficient $\rho_i$ belongs to this set is $\Bcal$. We assume without loss of generality that $\rho_1=0$, so $\Bcal$ contains the space of quasi-linear preferences.

Our preference domain has an intuitive structure. The agent has an underlying quasi-linear preference. However, the marginal disutility of paying rent is not uniform. The agent has a soft budget~$b_i$, i.e., an amount of money she expects to spend on housing. Paying one dollar above~$b_i$ entails a larger disutility than that of paying one dollar more still under~$b_i$. This reflects, for instance, that either the agent is budget constrained and needs to pay an interest rate on money paid above her budget, or that she needs to reallocate money from other needs like transportation or entertainment in order to pay rent over~$b_i$.

Individual payments should add up to  $m\in\R$, the house rent. An economy is described by the tuple $e:= (N,A,u,m)$.
 An allocation for $e$ is a pair $(r,\sigma)$ where $\sigma:N\rightarrow A$ is a bijection and $r:=(r_{a})_{a\in A}$ is such that $\sum_{a\in A}r_a=m$.  An allocation is \textit{envy-free for $e$} if no agent prefers the consumption of any other agent at the allocation. The set of these allocations is $F(e)$.


\section{Results}\label{Sec:Complexity}

\subsection{The main result}

For constants $a> 0$ and $b\in\R$,  we refer to the function $x\in\R\mapsto b+ax$ as a positive affine linear transformation.
We develop  a family of polynomial algorithms to calculate an allocation that maximizes (minimizes) the minimum (maximum) of  positive affine linear transformations of individual utility (or payments of rent) among all envy-free allocations when preferences belong to $\Bcal$.

\begin{theorem}\label{Th:main-takeout-complexity}Let $e:=(N,A,u,m)$ be such that $u\in\Bcal^N$. For $S\subseteq N$ and $C\subseteq A$, let $(f_{i})_{i\in S}$ and $(g_{c})_{a\in C}$ be lists of positive affine linear transformations. Then, there are polynomial algorithms with input size $n$ that compute an element in:
\begin{enumerate}
\item $\underset{(r,\sigma)\in F(e)}{\arg\max}\min_{i\in S}f_i(u_i(r_{\sigma(i)},\sigma(i)))$;
\item $\underset{(r,\sigma)\in F(e)}{\arg\min}\max_{i\in S}f_i(u_i(r_{\sigma(i)},\sigma(i)))$;
\item $\underset{(r,\mu)\in F(e)}{\arg\max}\min_{a\in C}g_a(r_a)$;
\item $\underset{(r,\mu)\in F(e)}{\arg\min}\max_{a\in C}g_a(r_a)$.
\end{enumerate}
\end{theorem}

We initially limit our presentation to our main application, i.e., the maxmin utility solution
\[\Rcal(N,A,u,m):=\underset{(r,\sigma)\in F(N,A,u,m)}{\arg\max}\min_{i\in N}u_i(r_{\sigma(i)},\sigma(i)).\]
We discuss the proof of the theorem in Sec.~\ref{Sec:proof-general}.

\citet{Eshwar-et-al-2018-Arxiv} introduced an algorithm that calculates an envy-free allocation when preferences are represented by piece-wise linear functions, a domain containing $\Bcal$.\footnote{A piece-wise linear utility function has the form $(r_i,a)\mapsto v_{ias}-\lambda_{ias}r_i$ for a collection of consecutive intervals $\{I_{as}^i\}_{a\in A,s\in S^i}$ that covers $\R$. Even though we  will state all our results for our domain $\Bcal$, they all generalize for the piece-wise linear domain when we require for each $i$, $|S^i|$ is polynomial in $n$ and marginal disutility of paying rent comes from a given fixed set.} When preferences are in $\Bcal$, so the number of different values of the budget violation index that an agent can report is $k$, their algorithm runs in $O(n^{k+c})$ for some $c>0$ \citep[Sec. 4.1]{Eshwar-et-al-2018-Arxiv}.\footnote{The leading algorithm in \citet{Eshwar-et-al-2018-Arxiv} solves the case in which slopes are integer powers of a given $(1+\varepsilon)$. With this they construct an approximate envy-free allocation for an arbitrary piece-wise linear economy. The approximation is polynomial in $1/\varepsilon$. Since our objective is to provide more expressive but simple preferences to the agents, we prefer to present only results for exact envy-free allocations for finitely many slopes. In this context calculating a maxmin envy-free allocation has a direct interpretation.}  Their algorithm does not produce an allocation satisfying further criteria.

Thus, the algorithms in Theorem~\ref{Th:main-takeout-complexity} significantly improve over both~\citet{Gal-et-al-2016-JACM} and\linebreak \citet{Eshwar-et-al-2018-Arxiv}'s algorithms because they apply to a non-linear domain of preferences and produce  allocations in specific selections of the envy-free set. We do pay a price in the generality of our result compared to \citet{Gal-et-al-2016-JACM}, for their algorithm applies to arbitrary linear transformations of the \emph{vector} of utilities. For instance, their algorithm can calculate an allocation that maximizes the summation of utilities for a particular group of agents among all envy-free allocations. We limit our scope to the narrower set of maxmin or minmax individual utility or individual rent selections of the envy-free set. Arguably there is no significant loss of our approach, for we still cover the selections that are actually used in practice. Furthermore, our restriction is necessary to leverage a topological property of these selections, rent monotonicity, that was never taken advantage of before from an algorithmic perspective.\footnote{Rent monotonicity is satisfied only by selections that can be written as maxmin or minmax operators of utility or rent \citep{Velez-2016-Rent}.}

At a high level, we borrow from \citet{Gal-et-al-2016-JACM} the representation of our problems as linear programs, and from \citet{Eshwar-et-al-2018-Arxiv} the strategy of first solving the problem for a rent high enough in which budget constraints are all violated and then proceeding recursively rebating amounts of rent. Our algorithm significantly differs in that we are able to maintain control of the location of our iterations with respect to the solution we want to calculate. Rent monotonicity requires that, as the rent of the house increases, one recommends allocations in which the rent of each room is higher and the welfare of each agent is lower \citep{Velez-2016-Rent}. Because of this, in the $n$-dimensional utility and rent spaces one can see each of these selections as a strictly monotone path that is parameterized by the aggregate rent to collect. For instance, the utility of an agent and the rent of each room in each element of $\Rcal(N,A,u,m)$ is the same for all the members of this set and are strictly monotone functions of $m$ \citep[see][]{Velez-2017-Survey}. Our approach is to rebate rent and stay on this path. Due to the piece-wise linearity of our domain, our algorithm may temporarily deviate from this objective. However, when this happens, we are able to increase the rent to collect and return to it without losing significant progress towards our goal. Remarkably, this non-monotone approach does stop in polynomial time.

To provide the reader clarity of the significance of our contribution, we proceed in three steps. First, we describe the economy of rebates and reshuffles, a formalism that we introduce and is essential in our construction. Then, we describe how \citet{Eshwar-et-al-2018-Arxiv}'s algorithm works. Finally, we introduce our algorithms and analysis.

\subsection{The economy of rebates and reshuffles.}

The essential step of the algorithms that we will discuss and construct is the following. Starting from an envy-free allocation, decrease the rent of each room and reshuffle rooms so no-envy is preserved. In what follows we introduce notation and provide the intuition why this is always possible for a small enough rebate, the so called Perturbation Lemma \citep{A-D-G-1991-Eca,Eshwar-et-al-2018-Arxiv}. This notation and intuition is essential to later introduce our algorithms for the more general problem.

Suppose that the vector of rents at the starting allocation is $r:=(r_a)_{a\in A}$. Since preferences are in $\Bcal$, for small  $\varepsilon>0$, the difference in utility between bundles $(r_a,a)$ and $(r_a-\varepsilon,a)$ is a linear function of $\varepsilon$. To preserve correctness of our process, we need to identify the value of these slopes and account for their eventual changes as one rebates rent.

\begin{definition}\rm For each $u\in \Bcal^N$ and $r\in\R^A$, let
\begin{enumerate}\item $SB^u(r):= \{(i,a)\in N\times A:r_a>b_i\}$;
\item\[\nu_{ia}(u,r):=\left\{\begin{array}{ll} v^i_a+\rho_ib_i&\textit{ if }(i,a)\in SB^u(r),\\v^i_a&\textrm{ otw;}\end{array}\right.\textrm{ and }\]
\[\lambda_{ia}(u,r):=\left\{\begin{array}{ll} 1+\rho_i&\textit{ if }(i,a)\in SB^u(r),\\1&\textrm{ otw.}\end{array}\right.\]

\end{enumerate}
\end{definition}
Note that for each $i\in N$ and $a\in A$, $\nu_{ia}(u,r)$ and $\lambda_{ia}(u,r)$ are such that for some $\varepsilon>0$ and each $r^\delta_a\in[r_a,r_a-\varepsilon]$, $u_i(r_a^\delta,a)=\nu_{ia}(u,r)-\lambda_{ia}(u,r)r_a^\delta$.

\medskip
Consider now a starting envy-free allocation $(r,\sigma)\in F(N,A,u,m)$. A rebate of rent is a vector  $x\in\R^A_{++}$ and a reshuffle is a bijection $\mu:N\rightarrow A$.   At rebate and reshuffle $(x,\mu)$ at~$r$, agent~$i$ receives bundle $(r_{\mu(i)}-x_{\mu(i)},\mu(i))$. Utility function~$u$ induces a utility function on rebates and reshuffles at $r$, given by $(x_a,a)\mapsto u_i(r_a-x_a,a)$.

Now, suppose that agent $i$ is indifferent between bundles $(r_a,a)$ and $(r_b,b)$. Then her preferences between rebates with these rooms are linear for a neighborhood of zero. That is, for each small enough rebate $x$, $(r_a-x_a,a)$ is at least as good as $(r_b-x_b,b)$ if and only if $\lambda_{ia}(u,r)x_a\geq \lambda_{ib}(u,r)x_b$. Note that each agent is indifferent among the best bundles in $\{(r_a,a):a\in A\}$ and this set includes her assignment at $(r,\sigma)$. Each agent may not be indifferent among all bundles in $(r,\sigma)$. Thus the economy of rebates and reshuffles may not be linear. However, if we assign linear utilities to bundles with rooms that are not in the best bundles at $(r,\mu)$, the economy becomes linear. As we will see, this greatly helps in our endeavor if we carefully choose the linearization and aim to rebate only a small amount of money.

\begin{definition}\rm Let $u\in \Bcal^N$, $r\in\R^A$ for which there is $(r,\sigma)\in F(N,A,u,m)$, and  $\Lambda>0$.

\begin{enumerate}\item For each $i\in N$ and  $a\in A$, let $\lambda_{ia}(u,r,\Lambda):=\lambda_{ia}(u,r)$  if $u_i(r_{\sigma(i)},\sigma(i))=u_i(r_a,a)$, and  $\lambda_{ia}(u,r,\Lambda):=\Lambda$ otherwise.\footnote{Note that $\lambda_{ia}(u,r,\Lambda)$ is well defined, i.e., it is invariant for the choice of $\sigma$ as long as $(r,\sigma)\in F(N,A,u,m)$.}

\item The $\Lambda$-modified rebates and reshuffles economy for $u$ at $r$ is an economy in which consumption space is $\R_{++}\times A$, and preferences are represented by:  for each $i\in N$ and $x_a>0$, $(x_a,a)\mapsto \lambda_{ia}(u,r,\Lambda)x_a$. 

\end{enumerate}
\end{definition}

Let $u\in \Bcal^N$ and $(r,\sigma)\in F(N,A,u,m)$. A $\Lambda$-modified rebates and reshuffles economy for $u$ at $r$ is isomorphic to a quasi-linear rent allocation economy in which preferences are strictly increasing in money, i.e., an economy in which the alternatives space is $\R\times A$ and agents have preferences $(y_a,a)\mapsto \hat u(y_a,a):=\log \lambda_{ia}(u,r,\Lambda)+y_a$. Thus, we can use all the power of the results for the quasi-linear domain in this economy. In particular, let $\mu$ be an assignment that maximizes $\sum_{i\in N}\log\lambda_{i\mu(i)}(u,r,\Lambda)$ and let $\varepsilon\in\R$. There is $(y,\mu)\in F(N,A,\hat u,\varepsilon)$ \citep{A-D-G-1991-Eca}. Since the economy is quasi-linear, for each $\eta\in\R$, $((y_a+\eta)_{a\in A},\mu)\in F(N,A,\hat u,\varepsilon+n\eta)$. Thus, by the Intermediate Value Theorem, for each $\delta>0$ we can construct an allocation in the linearized rebates and reshuffles economy $(y,\mu)\in F(N,A,\hat u,\sum_{a\in A}y_a)$ such that $\sum_{a\in A}\exp(y_a)=\delta$.

Our aim is to guarantee that $(r-\exp(y),\mu)\in F(N,A,u,m-\delta)$. Clearly, this is so if we guarantee two conditions: $(i)$ $\mu$ assigns each agent one object in the best bundles at $(r,\sigma)$; and $(ii)$ $\delta$ is small enough so for each $i\in N$, if $u_i(r_{\sigma(i)},\sigma(i))>u_i(r_{a},a)$, then $u_i(r_{\sigma(i)},\sigma(i))>u_i(r_{a}-\delta,a)$.

Continuity of preferences guarantees that $\delta$ can be chosen so $(ii)$ holds. The following lemma states that $\Lambda$ can be chosen small enough so $(i)$ holds.

\begin{lemma}\label{Lm:boundLambda} Let $u\in \Bcal^N$ and $r\in\R^A$ for which there is $(r,\sigma)\in F(N,A,u,m)$. There is  $\Lambda>0$, such that for each $0<\Lambda'\leq\Lambda$, if $\mu$ is a solution to
\[\max_{\gamma:N\rightarrow A,\ \gamma\textrm{ a bijection}}\sum_{i\in N}\log(\lambda_{i\gamma(i)}(u,r,\Lambda')) ,\]
then for each $i\in N$ and each $\sigma$ for which $(r,\sigma)\in F(N,A,u,m)$, $u_i(r_{\mu(i)},\mu(i))=u_i(r_{\sigma(i)},\sigma(i))$.
\end{lemma}
Proofs omitted in the body of the paper are presented in the Appendix.

The following graphs allow us to summarize the discussion above in one lemma.\footnote{Lemma~\ref{Lem:Perturbation} is a generalization of the Perturbation Lemma in \cite{A-D-G-1991-Eca} that determines the existence, but does not identify an assignment with which the rebate and reshuffle can proceed.}

\begin{definition}[\citealp{Eshwar-et-al-2018-Arxiv}]\rm For $u\in\Bcal^N$ and $r\in \R^A$ for which there is $(r,\sigma)\in F(N,A,u,m)$, let $\Fcal(r):=(N,A,E)$ be the bipartite graph were $(i,a)\in E$ if $u_i(r_{\sigma(i)},\sigma(i))=u_i(r_a,a)$; and $\Fcal^u(r)$ the weighted version of $\Fcal(r)$ were for each $(i,a)\in E$, $w(i,a):= \log \lambda_{ia}(u,r)$.
\end{definition}

\begin{lemma}[Perturbation Lemma; \citealp{Eshwar-et-al-2018-Arxiv}]\label{Lem:Perturbation} Let $(r,\sigma)\in F(N,A,u,m)$. Suppose that  $u$ is piece-wise linear and $\mu$ is a maximum weight perfect matching in $\Fcal^u(r)$. Then, there is $\varepsilon>0$  such that for each $\delta\in[0,\varepsilon]$, there is $(r^\delta,\mu)\in F(N,A,u,m-\delta)$ such that for each $a\in A$, $r^\delta_a<r_a$.\footnote{The discussion before the statement of the lemma is almost a detailed proof of it. More formally, Lemma~\ref{Lem:Perturbation-hmaxmin}, which we prove in detail in the Appendix, implies Lemma~\ref{Lem:Perturbation}.}
\end{lemma}


\subsection{\citet{Eshwar-et-al-2018-Arxiv}'s algorithm.}

\citet{Eshwar-et-al-2018-Arxiv} leverage Lemma~\ref{Lem:Perturbation} to construct a polynomial algorithm to find an envy-free allocation in a piece-wise linear economy in which the slopes in the different intervals comes from a finite set say of cardinality $k$. The essential step in this task is the following. Given $\eta>0$, starting from an envy-free allocation $(r,\sigma)\in F(N,A,u,m)$ with $u\in\Bcal^N$, rebate $\eta>0$ so no-envy and budget violations are preserved.  For arbitrary continuous preferences, Lemma~\ref{Lem:Perturbation} itself does not solve this problem, i.e., stacking the $\varepsilon$s produced by the lemma, may not allow one to reach a rebate of $\eta$ \citep{Alkan-1989-EJPE}. \citet{Eshwar-et-al-2018-Arxiv}'s breakthrough is to realize that, in the piecewise linear domain, this can be done by concatenating the solution of the following linear program for a maximal weight perfect matching in $\Fcal^u(r)$, $\mu$:
\begin{equation}\begin{array}{rll}\max_{t\in\R^A} &\sum_{a\in A}t_a&\\
s.t.:&t_a\leq r_a&\forall a\in A\\
&\nu_{i\mu(i)}(u,r)-\lambda_{i\mu(i)}(u,r)t_{\mu(i)}\geq \nu_{i\mu(j)}(u,r)-\lambda_{i\mu(j)}(u,r)t_{\mu(j)}&\forall\{i,j\}\subseteq N
\\
&t_a\leq b_i &\forall(i,a)\in SB^u(t)
\\
&\sum_{a\in A}t_a\geq\sum_{a\in A}r_a-\eta.&
\end{array}\label{Eq:Els-LP}\end{equation}
This problem is feasible because $r$ is in the option set. Thus, let $(t,\sigma)$ be the allocation associated with its solution. Since $t$ satisfies the second set of constraints in~(\ref{Eq:Els-LP}), the allocation is in $F(N,A,u,m-\varepsilon)$ for some $0<\varepsilon\leq \eta$. Thus, one of the following is true:  (i) the solution is in $F(N,A,u,m-\eta)$; (ii) the solution hit one of the budget constraints; or (iii) $\sigma$ is not a maximal weight perfect matching in $\Fcal^u(t)$. (If these three conditions simultaneously fail, by Lemma~\ref{Lem:Perturbation}, $t$ would not be a solution to the problem.) If one repeats this starting from the previously computed allocation,
eventually case (i) or (ii) happen in $O(n^{k-1})$, because since there are $k$ values for the slopes of the intervals, the maximum number of values for a maximal weight perfect matching in $\Fcal^u(s)$ for any $s\in\R^A$ is bounded above by $(n+1)^{k-1}$.

Thus, for $k$ constant, one can construct a polynomial algorithm that stops in $O(n^{k+c}\varsigma)$, where $\varsigma$ is the number of intervals in the piece-wise linear representation of preferences. In our case, with preferences in $\Bcal$, $\varsigma$ is bounded above by $2n^2$. Thus, the algorithm runs in $O(n^{k+c})$.

\subsection{Directed search within the envy-free set.}

A solution to (\ref{Eq:Els-LP}) leads us to a ``random'' envy-free allocation. We would like to optimize some further criteria within the envy-free set. The following algorithms achieve this. They calculate in polynomial time an allocation in $\Rcal(N,A,u,m)$ for $u\in\Bcal^N$.
\begin{algorithm}[t]
\SetKwData{Left}{left}\SetKwData{This}{this}\SetKwData{Up}{up}
\SetKwFunction{Union}{Union}\SetKwFunction{FindCompress}{FindCompress}\SetKwInOut{Input}{Input}\SetKwInOut{Output}{Output}
\Input{$(N,A,u,m)$ were $u\in\Bcal^N$ is associated with $(v_{ia})_{i\in N,a\in A}$, $(b_i)_{i\in N}$, and for each $i\in N$, $\rho_i\in\{\rho_1,...,\rho_k\}$;}
\Output{$M\geq m$ and an allocation in $\Rcal(N,A,u,M)$;}
{For each $i\in N$ and each $a\in A$, let $V_{ia}:= (v_{ia}+\rho_i b_i)/(1+\rho_i)$ and $\tilde u_{i}$ the function
$(x_a,a)\mapsto \tilde u_i(x_a,a):=(1+\rho_i)(V_{ia}-x_a)$}\;\label{Step:velez1-1}
{Let $m':= n\left(\max_{i\in N,\{a,b\in\}\subseteq A}V_{ib}-V_{ia}+\max_{j\in N}b_j\right)$}\;\label{Step:velez1-2}
{Let $\sigma$ be an assignment that maximizes $\sum_{i\in N}V_{i\sigma(i)}$}\;\label{Step:max-mu-velez1}
{Solve
\[\begin{array}{rll}\max_{R,r\in\R^A} &R&\\
s.t.:&R\leq \tilde u_i(r_{\sigma(i)},\sigma(i))&\forall i\in N\\
&V_{i\sigma(i)}-r_{\sigma(i)}\geq V_{i\sigma(j)}-r_{\sigma(j)}&\forall\{i,j\}\subseteq N
\\
&\sum_{a\in A}r_a=\max\{m,m'\}&
\end{array}\]
and let $r\in\R^A$ be a solution to this LP\;\label{Step:LP-velez1}}
{Returns $M:= \sum_{a\in A}r_a$ and $(r,\sigma)$}\;
\caption{Initializes search of maxmin utility envy-free allocation}
\label{alg:Velez1}
\end{algorithm}

Algorithm~\ref{alg:Velez1} initializes our search by looking for an allocation for a rent~$M$ that is high enough so we make sure that all budget constraints will be violated for each agent, for each consumption of the agent or the other agents, at each possible envy-free allocation for $(N,A,u,M)$. For such a high rent preferences are quasi-linear in the range of the consumption space that contains all envy-free allocations for $(N,A,u,M)$. Thus, one can find a maxmin allocation by essentially using \cite{Gal-et-al-2016-JACM}'s algorithm for this  quasi-linear preference (lines~\ref{Step:max-mu-velez1}-\ref{Step:LP-velez1}).

\begin{theorem}\label{Thm:analysis-alg-vel-1} Algorithm~\ref{alg:Velez1} stops in polynomial time. Given input $(N,A,u,m)$ where $u\in\Bcal$, the algorithm returns $M$ and $(r,\sigma)$ such that $M\geq m$ and $(r,\sigma)\in \Rcal(N,A,u,M)$.
\end{theorem}

\begin{proof}Lines \ref{Step:velez1-1} and \ref{Step:velez1-2} are direct definitions. Line 3 is well known to be polynomial in $n$. For line \ref{Step:LP-velez1} we should note that it is indeed a linear program, because all $\tilde u$s are linear in $r$ (at this point $\sigma$ is fixed). It is feasible because given a quasi-linear economy and an assignment $\sigma$ that maximizes the summation of values, there is always an envy-free allocation for that economy with that assignment \citep{A-D-G-1991-Eca}. Because of the second and third sets of constraints in the program, the feasible set is compact. Thus, the program has a solution. Since the number of constraints in the program is polynomial in $n$, the complexity of solving this problem is known to be polynomial in $n$. Let $M$ and $(r,\sigma)$ be the output of the algorithm. Clearly, $M\geq m$. We claim that $(r,\sigma)\in\Rcal(N,A,u,M)$. Since for each $i\in N$, $(1+\rho_i)>0$,  profile $\tilde u$ is ordinally equivalent to the quasi-linear profile with values $V$. Thus, the solution of the linear program in line \ref{Step:LP-velez1} solves $\max_{r\in\R^A:(r,\sigma)\in F(N,A,\tilde u, M)}\min_{l\in L}\tilde u_i(r_{\sigma(i)},\sigma(i))$.

Let $(t,\mu)\in \Rcal(N,A,\tilde u, M)$.  Since $\tilde u$ is ordinally equivalent to a quasi-linear preference and\linebreak $\{(r,\sigma),(t,\mu)\}\subseteq F(N,A,\tilde u,M)$, $(t,\sigma)\in F(N,A,\tilde u,M)$ \citep{Svensson-2009-ET}. Moreover, by \citet[Lemma 3,][]{A-D-G-1991-Eca}, for each $i\in N$, $\tilde u_i(t_{\mu(i)},\mu(i))= \tilde u_i(t_{\sigma(i)},\sigma(i))$. Thus, $\min_{l\in L}\tilde u_i(t_{\sigma(i)},\sigma(i))=\min_{l\in L}\tilde u_i(t_{\mu(i)},\mu(i))$. Since $\min_{l\in L}\tilde u_i(r_{\sigma(i)},\sigma(i))\geq\min_{l\in L}\tilde u_i(t_{\sigma(i)},\sigma(i))$, then $(r,\sigma)\in\Rcal(N,A,\tilde u,M)$.

We claim that $F(N,A,\tilde u,M)=F(N,A,u,M)$ and for each $(t,\mu)\in F(N,A,u,M)$, $u(t_{\mu(i)},\mu(i))=\tilde u(t_{\mu(i)},\mu(i))$. Let $t\in\R^A$ be such that, $\sum_{a\in A}t_a=M$ and $\mu:N\rightarrow A$ a bijection. Then, there is $i\in N$ such that $t_i\geq \left(\max_{i\in N,\{a,b\in\}\subseteq A}V_{ib}-V_{ia}+\max_{j\in N}b_j\right)$. Then,  $t_i\geq b_i$. Thus, $u_i(t_i,\mu(i))=\tilde u_i(t_i,\mu(i))=(1+\rho_i)(V_{i\mu(i)}-t_i)$. Now, $V_{i\mu(i)}-t_i\leq V_{i\mu(i)}-\max_{i\in N,\{a,b\in\}\subseteq A}V_{ib}-V_{ia}-\max_{j\in N}b_j$. Thus, for each $a\in A$, $V_{i\mu(i)}-t_i\leq V_{i\mu(i)}-(V_{i\mu(i)}-V_{ia})-\max_{j\in N}b_j=V_{ia}-\max_{j\in N}b_j$. Thus, for each $a\in A$, $u_i(t_i,\mu(i))=\tilde u_i(t_i,\mu(i))\leq u_i(\max_{j\in N}b_j,a)=\tilde u_i(\max_{j\in N}b_j,a)$. Suppose that $(t,\mu)\in F(N,A,\tilde u,M)$. Thus, for each $a\in A$, $t_a\geq \max_{j\in N}b_j$, for otherwise  $\tilde u_i$ will envy the agent who receives $a$ at $(t,\mu)$. Thus, for each pair $\{i,j\}\subseteq N$, $u_i(t_{\mu(j)},\mu(j))=\tilde u_i(t_{\mu(j)},\mu(j))$. Since $(t,\mu)\in F(N,A,\tilde u,M)$, for each pair $\{i,j\}\subseteq  N$, $u_i(t_{\mu(i)},\mu(i))\geq u_i(t_{\mu(j)},\mu(j))$. Thus, $(t,\mu)\in F(N,A,u,M)$. Suppose then that $(t,\mu)\in F(N,A,u,M)$. A symmetric argument shows that $(t,\mu)\in F(N,A,\tilde u,M)$.

Since $F(N,A,\tilde u,M)=F(N,A,u,M)$; for each $(t,\mu)\in F(N,A,u,M)$ and each $i\in N$, $u(t_{\mu(i)},\mu(i))=\tilde u(t_{\mu(i)},\mu(i))$, and $(r,\sigma)\in \Rcal(N,A,\tilde u,M)$, we have that  $(r,\sigma)\in\Rcal(N,A,u,M)$.
\end{proof}

It is worth noting that there is a subtle choice in the construction of Algorithm~\ref{alg:Velez1}. We first identify a rent that is high enough to guarantee all budget constraints are violated in our original economy for each envy-free allocation. Then, we construct a maxmin utility allocation for this particular rent. One may be tempted to simply consider the quasi-linear economy that coincides with our economy when budgets are violated, construct a maxmin allocation for this economy for an arbitrary rent, and then ``slide'' it so all budgets are violated. This approach does not work because we need a maxmin allocation with respect to a linear transformation of these quasi-linear utilities. These selections may not be invariant to uniform translations of money.

If Algorithm~\ref{alg:Velez1} returns $M=m$, we have actually computed an element of $\Rcal(N,A,u,m)$. Thus, we need to continue our search only when this algorithm returns $M>m$ and $(r,\sigma)\in \Rcal(N,A,u,M)$. Algorithm~\ref{alg:Velez2} does so. This algorithm shares some of its philosophy with \citet{Eshwar-et-al-2018-Arxiv}'s. At a given state in which an allocation in $\Rcal(N,A,u,m')$ with $m'>m$ has been calculated, it  reshuffles rooms and rebates rent by solving (\ref{EQ:LP1}), an LP that maximizes the minimum value of the $u_i$s constrained by no-envy and budget regime changes.

As in \citet{Eshwar-et-al-2018-Arxiv}'s algorithm, solving (\ref{EQ:LP1}) gets us closer to collecting exacly rent $m$. The solution to this problem, $t^s$, may be such that $(t^s,\sigma^s)\not\in \Rcal(N,A,u,\sum_{a\in A}t_a^s)$, however. The issue is that (\ref{EQ:LP1}) has constraints additional to the envy-free ones. In particular, the $SB$ constraints of this problem may bind at its solution.\footnote{The rebate constraints, i.e., $t^s_a\leq r^{s-1}_a$s, never bind at a solution to (\ref{EQ:LP1}).} Thus, since our objective is a maxmin utility envy-free allocation, this step may lose the maxmin property. When this is so, i.e., line \ref{EQ:r^s-velez2} is reached, we need to correct the situation. We do so by grabbing the value of (\ref{EQ:LP1}), $R^s$, and increasing rents again constrained by no-envy and maxmin utility $R^s$, i.e.,  by solving (\ref{EQ:LP2}). It turns out that the solution to (\ref{EQ:LP2}) fits the bill. First, it is a maxmin utility envy-free allocation with value~$R^s$. This is quite unexpected. At face value, problem (\ref{EQ:LP2}) returns a maxmin utility allocation constrained to assignment $\sigma^s$. This may, in principle, not be an unconstrained maxmin utility allocation. However, it turns out that, and here is where the subtlety of our analysis resides, when the algorithm overshoots rebating money with assignment $\sigma^s$, it also reveals that $\sigma^s$ admits a maxmin utility allocation for the highest rent possible before the algorithm ``fell'' from the maxmin utility path. Second, compared with~$t^s$, the solution to (\ref{EQ:LP2}) does not decrease any of the fundamental measures of progress in our algorithm. That is, at a solution of (\ref{EQ:LP2}), $r^s$, either  $\sigma^s$ is not a maximal weight perfect matching in $\Fcal^u(r^s)$ or a new budget constraint was released.

\begin{theorem}\label{Thm:analysis-alg-vel-2} Algorithm~\ref{alg:Velez2} stops in polynomial time. Given input $(N,A,u,m)$ where $u\in\Bcal$, its output belongs to $\Rcal(N,A,u,m)$.
\end{theorem}

\begin{algorithm}[H]
\SetKwData{Left}{left}\SetKwData{This}{this}\SetKwData{Up}{up}
\SetKwFunction{Union}{Union}\SetKwFunction{FindCompress}{FindCompress}\SetKwInOut{Input}{Input}\SetKwInOut{Output}{Output}

\Input{$(N,A,u,m)$, $u\in\Bcal^N$, $b:=(b_i)_{i\in N}$, $\rho:=(\rho_i)_{i\in N}\in\{\rho_1,...,\rho_k\}^N$, $M>m$, and $(r,\sigma)\in \Rcal(N,A,u,M)$}
\Output{an allocation in $\Rcal(N,A,u,m)$}
Initialize $s\gets 0$\;
Let $(r^s,\sigma^s):=(r,\sigma)$ and $R^s:=\min_{i\in N}u_i(r_{\sigma(i)},\sigma(i))$\;

\While{$\sum_{a\in A}r^s_a>m$}{
Update $s\gets s+1$\;\label{Step:update-s}
For each $i\in N$ and $a\in A$, let $\tilde u_i^s(t_a,a):= \nu^s_{ia}(u,r)-\lambda^s_{ia}(u,r)t_a$\;
 Let $\sigma^s$ be a maximum weight perfect matching in $\Fcal^u(r^{s-1})$\;\label{EQ:maximum-sigma-velez2}
Solve
\begin{equation}\begin{array}{rll}\max_{R,t^s\in\R^A} &R&\\
s.t.:&t^s_a\leq r^{s-1}_a&\forall a\in A
\\&R\leq \tilde u^s_i(t^s_{\sigma^s(i)},\sigma^s(i))&\forall i\in N\\
&\tilde u^s_i(t^s_{\sigma^s(i)},\sigma^s(i))\geq \tilde u^s_i(t^s_{\sigma^s(j)},\sigma^s(j))&\forall \{i,j\}\subseteq N
\\
&t^s_a\geq b_i&\forall (i,a)\in SB^u(r^{s-1})
\\
&\sum_{a\in A}t^s_a\geq m&
\end{array}\label{EQ:LP1}\end{equation}
and let $t^s\in\R^A$ and $R^s$ be a solution to this LP\;
\eIf{$(t^s,\sigma^s)\in \Rcal(N,A,u,\sum_{a\in A}t^s_a)$\label{Step:check-in-R}}
{$r^s\gets t^s$}
{\label{EQ:r^s-velez2}
Solve
\begin{equation}\begin{array}{rll}\max_{r^s\in\R^A} &\sum_{a\in A}r^s_a&\\
s.t.:&r^s_a\geq t^s_a&\forall a\in A
\\&R^s\leq \tilde u^s_i(r^s_{\sigma^s(i)},\sigma^s(i))&\forall i\in N\\
&\tilde u^s_i(r^s_{\sigma^s(i)},\sigma^s(i))\geq \tilde u^s_i(r^s_{\sigma^s(j)},\sigma^s(j))&\forall \{i,j\}\subseteq N
\end{array}\label{EQ:LP2}\end{equation}
and let $r^s$ be a solution to this LP\;}
}
Return $(r^s,\sigma^s)$\;
\caption{Calculates a maxmin envy-free allocation.}
\label{alg:Velez2}
\end{algorithm}

The first step in the proof of the theorem is to realize that the Perturbation Lemma can be strengthened to guarantee that perturbations can preserve the maxmin utility property.

\begin{lemma}[Maxmin perturbation Lemma]\label{Lem:Perturbation-hmaxmin} Let $(r,\sigma)\in \Rcal(N,A,u,m)$ such that $u\in \Bcal^N$ and $\mu$ a maximum weight perfect matching in $\Fcal^u(r)$. Then, there is $\varepsilon>0$  and a function $\delta\in[0,\varepsilon]\mapsto (r^\delta,\mu)\in \Rcal(N,A,u,m-\delta)$ such that $(r^0,\mu)=(r,\mu)$; and  for each pair $0\leq \delta<\eta\leq\varepsilon$, and each $i\in N$,  $u_i(r^\eta_{\mu(i)},\mu(i))>u_i(r^\delta_{\mu(i)},\mu(i))$, and for each $a\in A$, $r^\delta_a>r^\eta_a$.
\end{lemma}

Lemma~\ref{Lem:Perturbation-hmaxmin}  guarantees that the solution to~(\ref{EQ:LP1}) will decrease the aggregate rent we are collecting. Now, suppose that we are at some iteration of Algorithm~\ref{alg:Velez2} for $s>0$ in which we find that $(t^s,\sigma^s)\not\in\Rcal(N,A,u,\sum_{a\in A}t^s_a)$. Intuitively, the algorithm deviated from the monotone path (in utility space and rent space) that is determined by the maxmin utility envy-free solution. We argue now that by solving~(\ref{EQ:LP2}) we are able to return to this path without significant loss of progress. More precisely, our solution to~(\ref{EQ:LP2}), i.e., $(r^s,\sigma^s)$ is in  $\Rcal(N,A,u,\sum_{a\in A}r^s_a)$ and $R^s$ is the minimum utility at $(r^s,\sigma^s)$.

Suppose then that $R^s$, the value of problem~(\ref{EQ:LP1}), is the minimum utility at $(t^s,\sigma^s)$, an allocation that is not in the maxmin path. Let $R^*$ be the maximal minimum utility across agents, bounded above by $R^s$, at an envy-free allocation obtained by a rebate at $(r^{s-1},\sigma^{s-1})$, for assignment $\sigma^{s}$. If we want to keep our solution in the maxmin path, our best chance with assignment $\sigma^s$ is to return to an allocation with maxmin value~$R^*$. Obviously, $R^*\leq R^s$. It turns out that  $R^*=R^s$. Let $(r^*,\sigma^s)$ be an allocation associated with $R^*$ (in the set that defines $R^*$). If $R^*<R^s$, one can show that $(r^s,\sigma^s)$ is obtained by rebating money in each room at $(r^*,\sigma^s)$. Intuitively, this reveals that there was still room to rebate money at $(r^*,\sigma^s)$ without changing the assignment. The following key lemma states that this implies that $\sigma^s$ must be a maximal weight perfect matching in $\Fcal^u(r^*)$.
\begin{lemma}[Converse perturbation lemma]\label{Lem:ConversePerturbation}Let $u\in\Bcal^N$, $\varepsilon>0$, $(r,\sigma)\in F(N,A,u,m)$ and $(t,\sigma)\in F(N,A,u,m-\varepsilon)$ such that for each $a\in A$, $r_a>t_a$. Suppose that there is no $i\in N$ and $a\in A$ such that, $r_a>b_i>t_a$. Then, $\sigma$ is a maximal weight perfect matching in $\Fcal^u(r)$.
\end{lemma}

Thus, it is impossible that $R^*<R^s$, for otherwise by Lemma~\ref{Lem:Perturbation-hmaxmin}, $R^*$ would not be maximal. Consequently, in order to recover a maxmin allocation with value $R^s$, we can keep assignment $\sigma^s$ and increase rent from $t^s$, i.e., we can solve~(\ref{EQ:LP2}). Again by Lemma~\ref{Lem:Perturbation-hmaxmin}, either a budget constraint was just released at $(r^s,\sigma^s)$, or $\sigma^s$ is not a maximal weight perfect matching in $\Fcal^u(r^s)$. Thus, if by the time all budget constraints are released, which must happen in polynomial time, the algorithm has not stopped, it will do so after solving~(\ref{EQ:LP1}) only once more.

\begin{proof}[\textit{Proof of Theorem~\ref{Thm:analysis-alg-vel-2}}]We prove that the processes in each line of the algorithm are well defined and can be individually completed in polynomial time. Then we bound the number of times the while loop is visited.

For each $s\geq 1$, $(r^{s-1},\sigma^s)\in F(N,A,u,\sum_{a\in A}r^{s-1}_a)$, because it is either the input when $s=0$ or  the solution to either (\ref{EQ:LP1}) or (\ref{EQ:LP2}), which have no-envy constraints. As long as $(r^{s-1},\sigma^s)\in F(N,A,u,\sum_{a\in A}r^{s-1}_a)$, $r^{s-1}$ is in the feasible set of (\ref{EQ:LP1}).  Because of the first and last constraints, this set is compact. Thus, it has a solution. Since it is a linear program with a polynomial number of constraints, it can be computed in polynomial time.

Line~\ref{Step:check-in-R} can be completed in polynomial time, i.e., given $t^s\in\R^A$ and $\sigma^s:N\rightarrow A$ a bijection, $(t,\sigma)\in \Rcal(N,A,u,\sum_{a\in A}t_a)$ is verifiable in polynomial time. By \citet[Proposition 5.9][]{Velez-2017-Survey}, this problem is equivalent to check that the allocation is envy-free and that for a given directed graph with $n$ nodes, there is a path from each node to a given set of nodes.

Because of the first and second constraint in (\ref{EQ:LP2}), the feasible set in this program is compact.\footnote{Even if utility was not linear, this program is compact because $R^s$ is the value of (\ref{EQ:LP1}).} Since $t^s$ belongs to this feasible set, the program has a solution.

We claim that the algorithm stops. Moreover, if it returns $(r^s,\sigma^s)$, $s$ is bounded by $n^{k+2}$. Thus, the algorithm runs in $O(n^{k+c})$ for some $c>2$.

Suppose that we update~$s$ in line~\ref{Step:update-s}, $\sum_{a\in A}r^{s-1}_a>m$,  $(r^{s-1},\sigma^{s-1})\in \Rcal(N,A,u,\sum_{a\in A}r^{s-1}_a)$, and $R^{s-1}=\min_{i\in N}u_i(r_{\sigma^{s-1}(i)},\sigma^{s-1}(i))$. Let $t^s$ be the solution to~(\ref{EQ:LP1}). Then, $(t^s,\sigma^s)\in F(N,A,u,\sum_{a\in A}t^s_a)$. Thus, $\sigma^s$ is a perfect matching in $\Fcal^u(t^s)$. If $\sum_{a\in A}t^s_a=m$, the algorithm terminates. Suppose then that $\sum_{a\in A}t^s_a>m$. We claim that  $(r^s,\sigma^s)\in \Rcal\left(N,A,u,\sum_{a\in A}r^s_a\right)$,  $R^s=\min_{i\in N}u_i(r^s_{\sigma^s(i)},\sigma^s(i))$, and~$r^s$ is a solution to~(\ref{EQ:LP1}). This is obviously so when  $(t^s,\sigma^s)\in \Rcal(N,A,u,\sum_{a\in A}t^s_a)$ and thus $r^s=t^s$.

Suppose then that $(t^s,\sigma^s)\not\in \Rcal(N,A,u,\sum_{a\in A}t^s_a)$. Thus, $r^s$ is a solution to~(\ref{EQ:LP2}). Let $R^*$ be the maximum of
\begin{equation}\left\{R\leq R^s:\exists(t,\sigma^s)\in \Rcal\left(N,A,u,\sum_{a\in A}t_a\right),R=\min_{i\in N}u_i(t_{\sigma^s(i)},\sigma^s(i))\right\}.\label{EQ:setR}\end{equation}
By Lemma~\ref{Lem:Perturbation-hmaxmin}, the set above is non-empty. Since preferences are continuous, the set is also closed (this follows from \citet[Proposition 2-2,][]{Velez-2016-Rent} and \citet[Decomposition Lemma,][]{A-D-G-1991-Eca}). Thus, $R^*$ is well-defined and $R^*>R^{s-1}$.  Let $(r^*,\sigma^s)\in \Rcal\left(N,A,u,\sum_{a\in A}r^*_a\right)$ be such that\linebreak $R^*=\min_{i\in N}u_i(r^*_{\sigma^s(i)},\sigma^s(i))$. We claim that $R^*=R^s$. Suppose by contradiction that $R^*<R^s$. Since $R^*>R^{s-1}$, by \citet[Theorem 1,][]{Velez-2016-Rent} and \citet[Decomposition Lemma,][]{A-D-G-1991-Eca}, for each $a\in A$, $r^{s-1}_a>r^*_a$.  By \citet[Proposition 5.9 and Lemma 5.7,][]{Velez-2017-Survey}, for each $a\in A$, $r^*_a>r^s_a$. Because of the first constraints in~(\ref{EQ:LP1}), this program rebates money from $(r^{s-1},\sigma^{s-1})$. Since the program is also constrained by budget regime changes, i.e., the $SB$ constraints, $B^u(r^*)=B^u(r^s)$. Thus, by Lemma~\ref{Lem:ConversePerturbation}, $\sigma^s$ is a maximal weight perfect matching in $\Fcal^u(r^*)$. By Lemma~\ref{Lem:Perturbation-hmaxmin}, $R^*$ is not the maximal element of~(\ref{EQ:setR}). This is a contradiction.

Since $R^*=R^s$, we have that $(r^*,\sigma^s)\in \Rcal\left(N,A,u,\sum_{a\in A}r^*_a\right)$ and  $R^s=\min_{i\in N}u_i(r^*_{\sigma^s(i)},\sigma^s(i))$. Since for each $i\in N$, $u_i(t^s_{\sigma^s(i)},\sigma^s(i))\geq R^s$, by \citet[Proposition 5.9 and Lemma 5.7,][]{Velez-2017-Survey}, for each $a\in A$, $t^s_a\geq r^*_a$.  Thus, $r^*$ is in the feasible set of~(\ref{EQ:LP2}). Since $r^s$ is a solution of~(\ref{EQ:LP2}), it is also in its feasible set. Thus, for each $i\in N$, $R^s\leq \min_{i\in N}u_i(r^*_{\sigma^s(i)},\sigma^s(i))$. By \citet[Proposition 5.9 and Lemma 5.7,][]{Velez-2017-Survey}, for each $a\in A$, $r^s_a\geq r^*_a$. Since $r^s$ is a solution of~(\ref{EQ:LP2}), $r^*=r^s$. Thus, $(r^s,\sigma^s)\in \Rcal\left(N,A,u,\sum_{a\in A}r^s_a\right)$ and  $R^s=\min_{i\in N}u_i(r^s_{\sigma^s(i)},\sigma^s(i))$. Thus, $r^s$ is also a solution to~(\ref{EQ:LP1}).

We claim that it must be the case that either there is $(i,a)\in SB^u(r^{s-1})$ such that $r^s_a=b_i$, or $\sigma^s$ is not a maximal weight perfect matching in $\Fcal^u(r^s)$. Suppose by contradiction that for each $(i,a)\in SB^u(r^{s-1})$, $r^s_a>b_i$, and $\sigma^s$ is a maximal weight perfect matching in $\Fcal^u(r^s)$. By Lemma~\ref{Lem:Perturbation-hmaxmin}, there is $\delta>0$ and $t^\delta\in\R^A$ such that $(t^\delta,\sigma^s)\in \Rcal(N,A,u,m-\delta)$; for each $a\in A$, $t_a^\delta<r^s_a\leq r^{s-1}_a$; for each $(i,a)\in SB^u(r^{s-1})$, $t^\delta_a>b_i$; and $\sum_{a\in A}t^\delta_a>m$; and for each $i\in N$, $u_i(t^\delta_{\sigma^s(i)},\sigma^s(i))>u_i(r^s_{\sigma^s(i)},\sigma^s(i))$. Since for each $(i,a)\in B^u(r^{s-1})$, $t^\delta_a>b_i$, then for each $i\in N$, $\tilde u_i(t^\delta_{\sigma^s(i)},\sigma^s(i))>\tilde u_i(r^s_{\sigma^s(i)},\sigma^s(i))$. Thus, $\min_{i\in N}\tilde u^s_i(r^s_{\sigma^s(i)},\sigma^s(i))<\min_{i\in N}\tilde u^s_i(t^\delta_{\sigma^s(i)},\sigma^s(i))$. Thus, $r^s$ is not a solution to~(\ref{EQ:LP1}). This is a contradiction.

Thus, each time that $s$ is updated, either $|SB^u(r^s)|<|SB^u(r^{s-1})|$ or the weight of $\sigma^s$ in $\Fcal^u(r^{s-1})$ is greater than the weight of $\sigma^{s-1}$ in $\Fcal^u(r^{s-2})$. There are at most $n^2$ elements in $SB^u(r^0)$ and at most $(n+1)^{k-1}$ values for the weight of a perfect matching. Thus, the algorithm either stops or reaches a state in which $SB^u(r^s)=\emptyset$ in $O(n^{k+c})$ for some $c>0$. If $SB^u(r^s)=\emptyset$, the algorithm returns the solution to (\ref{EQ:LP1}) in the next while loop iteration. Indeed, without the SB constraints,  (\ref{EQ:LP1}) always returns a maxmin utility allocation. Now, each optimal assignment in a quasi-linear economy can substitute the assignment in any envy-free allocation preserving no-envy \citep{Svensson-2009-ET,Gal-et-al-2016-JACM} (see \citep{Procaccia-Velez-Yu-2018-AAAI} for a short proof). Since the slopes of $u$ are then invariant under a decrease of rent, if the while loop were going to be visited for $s+1$, then $\sigma^s$ would still be a maximal perfect matching in $\Fcal^u(r^s)$, and by Lemma~\ref{Lem:left-Perturbation-hmaxmin}, $t^s$ would not be a solution to~(\ref{EQ:LP1}).
\end{proof}

\subsection{Proof of Theorem~\ref{Th:main-takeout-complexity}}\label{Sec:proof-general}

Algorithms~\ref{alg:Velez1} and~\ref{alg:Velez2} can be modified to calculate the allocations for all the selections in Theorem~\ref{Th:main-takeout-complexity}. Formally, a \textit{selection} (from the envy-free set) is a set valued function that associates with each economy a subset of envy-free allocations for it. The generic selection is $(N,A,u,m)\mapsto\Psi(N,A,u,m)\subseteq F(N,A,u,m)$. A selection $\Psi$ is \textit{essentially single-valued} if each agent is indifferent between each pair of allocations in $\Psi(N,A,u,m)$. It is well-known that if $\Psi$ is essentially single-valued, for a given room, say $a\in A$, the rent assigned to room $a$ is the same in each allocation in $\Psi(N,A,u,m)$ \citep{A-D-G-1991-Eca,Velez-2016-Rent}. An essentially single-valued selection $\Psi$ is \textit{rent monotone} if for each $i\in N$, the function $(N,A,u,m)\mapsto u_i(r_{\sigma(i)},\sigma(i))$ for $(r,\sigma)\in \Psi(N,A,u,m)$ is a strictly decreasing function of $m$. It is well-known that an essentially single-valued $\Psi$ is rent monotone if and only if the rent component functions $(N,A,u,m)\mapsto r_a$ where $(r,\sigma)\in \Psi(N,A,u,m)$, are strictly increasing functions of $m$ \citep{A-D-G-1991-Eca,Velez-2016-Rent}. A selection $\Psi$ is \textit{Pareto indifferent} if for each $(r,\sigma)\in\Psi(N,A,u,m)$ and each $(t,\mu)$ such that for each $i\in N$, $u_i(r_{\sigma(i)},\sigma(i))=u_i(t_{\mu(i)},\mu(i))$, we have that $(t,\mu)\in\Psi(N,A,u,m)$. Each selection in the statement of Theorem~\ref{Th:main-takeout-complexity} is essentially single-valued, rent monotone, and Pareto indifferent \citep{A-D-G-1991-Eca,Velez-2016-Rent}.

 First, consider a non-empty family of positive affine linear transformations $(f_{i})_{i\in S}$ for some $S\subseteq N$. One can calculate an allocation in
\[\underset{(r,\sigma)\in F(e)}{\arg\max}\min_{i\in S}f_i(u_i(r_{\sigma(i)},\sigma(i)))\]
by trivially modifying these algorithms as follows: In the LP in Algorithm~\ref{alg:Velez1} replace the maxmin constraints $R\leq\tilde u_i(\cdot),\,\forall i\in N$ with $R\leq f_i(\tilde u_i^s(\cdot)),\,\forall i\in S$; in Algorithm~\ref{alg:Velez2} replace $R\leq\tilde u_i^s(\cdot),\,\forall i\in N$ with $R\leq f_i(\tilde u_i^s(\cdot)),\,\forall i\in S$ in (\ref{EQ:LP1}), and replace $R^s\leq\tilde u_i^s(\cdot),\,\forall i\in N$ with $R\leq f_i(\tilde u_i^s(\cdot)),\,\forall i\in S$ in (\ref{EQ:LP2}). The analysis of correctness and complexity of the algorithms goes through unmodified. In particular, the only lemma that refers to $\Rcal$, Lemma~\ref{Lem:Perturbation-hmaxmin}, can be easily generalized to any essentially single-valued, rent monotone, and Pareto indifferent selection of the envy-free set.

Similarly, for a non-empty family of positive affine linear transformations $(g_a)_{a\in C}$ for some $C\subseteq A$, one can calculate an element of
\begin{equation}\underset{(r,\mu)\in F(e)}{\arg\min}\max_{a\in C}g_a(r_a)\label{EQ:minmax-rent}\end{equation}
as folllows. Replace the LP in Algorithm~\ref{alg:Velez1} with
\begin{equation}\begin{array}{rll}\min_{R,r\in\R^A} &R&\\
s.t.:&R\geq \tilde g_a(r_a)&\forall a\in C\\
&V_{i\sigma(i)}-r_{\sigma(i)}\geq V_{i\sigma(j)}-r_{\sigma(j)}&\forall\{i,j\}\subseteq N
\\
&\sum_{a\in A}r_a=\max\{m,m'\},&
\end{array}\label{EQ:LP3}\end{equation}
and replace (\ref{EQ:LP1}) with
\begin{equation}\begin{array}{rll}\min_{R,t^s\in\R^A} &R&\\
s.t.:&t^s_a\leq r^{s-1}_a&\forall a\in A
\\&R\geq g_a(r_a)&\forall a\in C\\
&\tilde u^s_i(t^s_{\sigma^s(i)},\sigma^s(i))\geq \tilde u^s_i(t^s_{\sigma^s(j)},\sigma^s(j))&\forall \{i,j\}\subseteq N
\\
&t^s_a\geq b_i&\forall (i,a)\in SB^u(r^{s-1})
\\
&\sum_{a\in A}t^s_a\geq m&
\end{array}\label{EQ:LP4}\end{equation}
and (\ref{EQ:LP2}) with
\begin{equation}\begin{array}{rll}\max_{r^s\in\R^A} &\sum_{a\in A}r^s_a&\\
s.t.:&r^s_a\geq t^s_a&\forall a\in A
\\&R^s\geq  g_a(r_a)&\forall a\in C\\
&\tilde u^s_i(r^s_{\sigma^s(i)},\sigma^s(i))\geq \tilde u^s_i(r^s_{\sigma^s(j)},\sigma^s(j))&\forall \{i,j\}\subseteq N
\end{array}\label{EQ:LP5}\end{equation}
The above modifications are rather obvious. It is very useful to emphasize why they do work, however. We will encounter next that the modification of the algorithms for the other two families of solutions in the theorem requires some extra thought.

To fix ideas suppose that $C=A$ and the $g_a$s are the identity functions. Thus, the modifications above lead to an allocation in which the maximal rent paid by some agent is minimized in the envy-free set, the minmax rent envy-free solution. Clearly,~(\ref{EQ:LP3}) produces a minmax rent envy-free allocation for the quasi-linear preferences that coincide with $u$ in the range for which all budgets are violated. Thus, it is again a viable seed for Algorithm~\ref{alg:Velez2}. Now consider~(\ref{EQ:LP4}). This LP is minimizing the maximal rent with some constraints for the economies in which rent is at least $m$. In other words, the problem is trying to rebate money  constrained by the limit to rebate up to the point in which aggregate rent is $m$. This is why its solution inches towards our goal. Similarly, LP~(\ref{EQ:LP5}) is trying to increase the aggregate rent to collect constrained by the maximal rent being $R^s$.  This is why its solution rights a possible overshoot by~(\ref{EQ:LP4}). The replication of our analysis with $\Rcal$ is perfunctory because~(\ref{EQ:minmax-rent}) defines an essentially single-valued, rent monotone, Pareto indifferent selection from the envy-free set.

Now  consider again a non-trivial family of positive  affine linear transformations $(f_i)_{i\in S}$ with $S\subseteq N$. Our objective is to calculate an element in
\begin{equation}\underset{(r,\sigma)\in F(e)}{\arg\min}\max_{i\in S}f_i(u_i(r_{\sigma(i)},\sigma(i))).\label{EQ:minmax-f}\end{equation}
A  trivial modification of the LP in Algorithm~\ref{alg:Velez1} in which one sets a minimization problem and replaces the maxmin constraints $R\leq\tilde u_i(\cdot)$ with the minmax constraints $R\geq f_i(\tilde u_i^s(\cdot))$ does produce a minmax envy-free allocation for the range in which budget constraints are violated. However, one finds a hurdle if one trivially transforms~(\ref{EQ:LP1}) into a minmax LP. The issue is that this problem is intended to rebate rent. If the objective becomes to minimize the maximal utility, the modified LP does not try to do this, because the maximal utility becomes lower as rent increases.

Thus, in order to find an element of~(\ref{EQ:minmax-f}) we need to rethink the whole structure of our approach. Instead of initially calculating an allocation for an aggregate rent that is large enough so the budget constraints are violated and then rebate rent, we need to do the opposite: Calculate an allocation for low enough rent so no budget constraint is violated and then increase rent. Modifying and analyzing Algorithm~\ref{alg:Velez1} is again perfunctory. Now that we are recursively increasing rent (last constraint in~(\ref{EQ:LP1})  flips to $\leq$), a minmax version of~(\ref{EQ:LP1}) and a min version of~(\ref{EQ:LP2})  work in the right direction. There are still two problems. First, the lemmas that guide our analysis of Algorithm~\ref{alg:Velez2} are not useful anymore, for they refer to rebates of rent. Second, we need to revise the choice  of the assignment for which these programs must be solved (Line~\ref{EQ:maximum-sigma-velez2} of Algorithm~\ref{alg:Velez2}). These two issues are related because our choice of this assignment is suggested by Lemma~\ref{Lem:Perturbation-hmaxmin}.

Suppose then that starting at some allocation $(r,\sigma)\in F(N,A,u,\sum_{a\in A}r_a)$ we intend to increase rent. Analogously to the economy of rebates and reshuffles, we can define an economy of surcharges and reshuffles as follows.

Given a preferce $u\in\Bcal$ let $\kappa_{ia}(u,r)$ be the absolute value of agent~$i$'s marginal disutility of an increase of rent on room $a$ at $r$ (this is the equivalent to $\lambda_{ia}(u,r)$ for an increase of rent). A surcharge of rent is a vector  $x\in\R^A_{++}$ and a reshuffle is a bijection $\mu:N\rightarrow A$.   At surcharge and reshuffle $(x,\mu)$ at~$r$, agent~$i$ receives bundle $(r_{\mu(i)}+x_{\mu(i)},\mu(i))$. Utility function~$u$ induces a utility function on surcharges and reshuffles at $r$, given by $(x_a,a)\mapsto u_i(r_a+x_a,a)$. Now, suppose that agent $i$ is indifferent between bundles $(r_a,a)$ and $(r_b,b)$. Then her preferences between surcharges with these rooms are negative linear for a neighborhood of zero. That is, for each rebate $x$ small enough (so budget regimes do not change), $(r_a+x_a,a)$ is at least as good as $(r_b+x_b,b)$ if and only if $-\kappa_{ia}(u,r)x_a\geq -\kappa_{ib}(u,r)x_b$.

Note that each agent is indifferent among the best bundles in $\{(r_a,a):a\in A\}$ and this set includes her assignment at $(r,\sigma)$. Each agent may not be indifferent among all bundles in $(r,\sigma)$. Thus the economy of surcharges and reshuffles may not be negative linear. Again, for $K>0$ we can assign utilities $-K(\cdot)$ to the bundles with rooms that are not in the best bundles at $(r,\mu)$, so the economy becomes so. Denote the absolute value of the slopes of these utilities by $\kappa(u,r,K)$.

Let $u\in \Bcal^N$ and $(r,\sigma)\in F(N,A,u,m)$. A $K$-modified surcharges and reshuffles economy for $u$ at $r$ is isomorphic to a quasi-linear rent allocation economy in which preferences are strictly \emph{decreasing} in money, i.e., an economy in which the alternatives space is $\R\times A$ and agents have preferences $(y_a,a)\mapsto \hat u(y_a,a):=-\log \kappa_{ia}(u,r,K)-y_a$. Thus, again we can use all the power of the results for the quasi-linear domain in this economy. In particular, let $\mu$ be an assignment that \emph{minimizes} $\sum_{i\in N}\log\kappa_{i\mu(i)}(u,r,K)$ and let $\varepsilon\in\R$. There is $(y,\mu)\in F(N,A,\hat u,\varepsilon)$ \citep{A-D-G-1991-Eca}. Since the economy is quasi-linear, for each $\eta\in\R$, $((y_a+\eta)_{a\in A},\mu)\in F(N,A,\hat u,\varepsilon+n\eta)$. Thus, by the Intermediate Value Theorem, for each $\delta>0$ we can construct an allocation in the linearized surcharges and reshuffles economy $(y,\mu)\in F(N,A,\hat u,\sum_{a\in A}y_a)$ such that $\sum_{a\in A}\exp(y_a)=\delta$.

Our aim is to guarantee that $(r+\exp(y),\mu)\in F(N,A,u,m+\delta)$. Clearly, this is so if we guarantee two conditions: $(i)$ $\mu$ assigns each agent one room in the best bundles at $(r,\sigma)$; and $(ii)$ $\delta$ is small enough so for each $i\in N$, if $u_i(r_{\sigma(i)},\sigma(i))>u_i(r_{a},a)$, then $u_i(r_{\sigma(i)},\sigma(i))>u_i(r_{a}+\delta,a)$.

Continuity of preferences guarantee that $\delta$ can be chosen so $(ii)$ holds. Again  $K$ can be chosen \emph{large} enough so $(i)$ holds, i.e., a symetric version of Lemma~\ref{Lm:boundLambda} holds. Thus, clearly the following symmetric version of Lemma~\ref{Lem:Perturbation} holds. We denote by $\Fcal^u_\kappa(r)$ the weighted version of $\Fcal(r)$ were for each $(i,a)\in E$, $w(i,a):= \log \kappa_{ia}(u,r)$.

\begin{lemma}[Left perturbation Lemma]\label{Lem:left-Perturbation} Let $(r,\sigma)\in F(N,A,u,m)$. Suppose that  $u$ is piece-wise linear and $\mu$ is a \underline{minimum} weight perfect matching in $\Fcal^u_\kappa(r)$. Then, there is $\varepsilon>0$  such that for each $\delta\in[0,\varepsilon]$, there is $(r^\delta,\mu)\in F(N,A,u,m+\delta)$ such that for each $a\in A$, $r^\delta_a>r_a$.\end{lemma}

By working on the surcharges and rebates economy instead of the rebates and reshuffles economy one also proves the following symmetric versions of Lemmas~\ref{Lem:Perturbation-hmaxmin} and~\ref{Lem:ConversePerturbation}. We omit the proof of these results, which can be completed with symmetric arguments to those in our proofs of Lemmas~\ref{Lem:Perturbation-hmaxmin} and~\ref{Lem:ConversePerturbation}.

\begin{lemma}[Monotone left perturbation Lemma]\label{Lem:left-Perturbation-hmaxmin} Let $\Psi$ be an essentially-single valued, rent monotone, Pareto indifferent selection from envy-free set. Suppose that $u\in \Bcal^N$ and $(r,\sigma)\in \Psi(N,A,u,m)$. Let $\mu$ be a minimum weight perfect matching in $\Fcal^u_\kappa(r)$. Then, there is $\varepsilon>0$ such that for each $\delta\in[0,\varepsilon]$ there is $(r^\delta,\mu)\in \Psi(N,A,u,m+\delta)$.
\end{lemma}

\begin{lemma}[Converse left perturbation lemma]\label{Lem:left-ConversePerturbation}Let $u\in\Bcal^N$, $\varepsilon>0$, $(r,\sigma)\in F(N,A,u,m)$ and $(t,\sigma)\in F(N,A,u,m-\varepsilon)$ such that for each $a\in A$, $r_a>t_a$. Suppose that there is no $i\in N$ and $a\in A$ such that, $r_a>b_i>t_a$. Then, $\sigma$ is a minimal weight perfect matching in $\Fcal^u_\kappa(t)$.
\end{lemma}

The modification of Algorithm~\ref{alg:Velez2} is completed by selecting in Line~\ref{EQ:maximum-sigma-velez2} a minimal weight perfect matching in $\Fcal^u(r^{s-1})$. The analysis follows then from similar arguments based on Lemmas~\ref{Lem:left-Perturbation-hmaxmin}  and~\ref{Lem:left-ConversePerturbation}.

Finally, the modification to compute an element of a maxmin rent  envy-free allocation (third family in the theorem) is then perfunctory.

\section{Discussion}\label{Sec:Discussion}

\textbf{Incentives}. A relevant question is the extent to which the algorithms we construct are manipulable. The incentives of envy-free rent division algorithms are relatively well understood. First, no such algorithm is  dominant strategies incentive compatible. This goes back to \citet{Green-Lafont-1979}, for envy-free rent allocations are budget balanced and Pareto efficient. Strikingly, for each mechanism that admits quasi-linear reports, there is at most one profile of preferences in which each agent's true report is a dominant strategy \citep{Velez-2017-Survey}. Better news are obtained for complete information non-cooperative incentives. If preference reports are required to be quasi-linear, the non-cooperative outcomes (limit Nash equilibrium outcomes) from the complete information manipulation games induced by any envy-free allocation algorithm are exactly the envy-free allocations for the true preference profile \citep{Velez-2015-JET,Velez-2017-Survey}. If agents are allowed to report arbitrary continuous preferences, and one had the means to calculate an envy-free allocation for such reports, there may be inefficient allocations that result as non-cooperative outcomes for this allocation process \citep{Velez-2015-JET}. This problem arises only when agents can report arbitrarily large marginal disutility of paying rent, however. If preferences are required to be piece-wise linear budget constrained with an upper bound on the marginal disutility of paying rent, the incentives in the quasi-linear domain are preserved. That is, when reports are required to be in $\Bcal$, the non-cooperative outcomes from the complete information manipulation games induced by any envy-free allocation algorithm are exactly the envy-free allocations for the true preference profile \citep{Velez-Axiv-BudgetIncentives}.

\textbf{Elicitation and calibration}. One can construct elicitation schemes for the preference parameters $(v^i_a)_{a\in A}$, $b_i$ and $\rho_i$ based on a strict interpretation of this domain as representing the actual preferences of the agents \citep{Velez-Axiv-BudgetIncentives}. In practice, one may want to deploy a version of this model in which agents are asked for their financial constraints on a coarse scale, say low, medium, and high. Since no-envy can be tested ex-post, it is plausible that one can calibrate values for these reports (based on experimental or field data) to maximize the performance of these mechanisms. These are open questions that are left for future research.

\bibliographystyle{ACM-Reference-Format}
\bibliography{ref-expressive-domains}

\section{Appendix}

\begin{proof}[\textit{Proof of Lemma~\ref{Lm:boundLambda}}]Let $\Lambda>0$ be such that for each $a\in A$ such that\linebreak $u_i(r_{\sigma(i)},\sigma(i))>u_i(r_a,a)$,
$\log \Lambda+(n-1)\max_{i\in N,a\in A}\lambda_{ia}<\sum_{i\in N}\lambda_{i\sigma(i)}$.
\end{proof}

The following lemma plays a key role in the proof of Lemmas~\ref{Lem:Perturbation-hmaxmin} and~\ref{Lem:ConversePerturbation}.

\begin{lemma}\label{Lem:induceFinlinear} Let $u\in\Bcal^N$, $\varepsilon>0$, $(r,\gamma)\in F(N,A,u,m)$, and $(t,\sigma)\in F(N,A,u,m-\varepsilon)$ such that $\sigma$ is a perfect matching in $\Fcal(r)$ and for each $a\in A$, $r_a>t_a$. Suppose that there is no $i\in N$ and $a\in A$ such that, $r_a>b_i>t_a$. Then, there is $\Lambda>0$ satisfying the property of Lemma~\ref{Lm:boundLambda} for $r$, such that
\[(r-t,\sigma)\in F\left(N,A,(\lambda_{ia}(u,r,\Lambda)(\cdot))_{i\in N,a\in A},\sideset{}{_{a\in A}}\sum r_a-t_a\right).\]
\end{lemma}

\begin{proof}[\textit{Proof of Lemma~\ref{Lem:induceFinlinear}}]Since $\sigma$ is a perfect matching in $\Fcal(r)$, $(r,\sigma)\in F(N,A,u,m)$.
Let $\Lambda>0$, $i\in N$, and $a\in A$ such that $u_i(r_{\sigma(i)},\sigma(i))=u_i(r_a,a)$. Thus, $\lambda_{ia}(u,r,\Lambda)=\lambda_{ia}(u,r)$. Since $(t,\sigma)\in F(N,A,u,m-\varepsilon)$, for each $\{i,j\}\subseteq N$, $u_i(t_{\sigma(i)},\sigma(i))\geq u_i(t_a,a)$. Thus, $u_i(t_{\sigma(i)},\sigma(i))-u_i(r_{\sigma(i)},\sigma(i))\geq u_i(t_a,a)-u_i(r_a,a)$. Since there is no $i\in N$ and $a\in A$ such that, $r_a>b_i>t_a$, this inequality can be written as $\lambda_{i\sigma(i)}(u,r)(r_{\sigma(i)}-t_{\sigma(i)})\geq \lambda_{ia}(u,r)(r_{a}-t_{a})$. Thus, $\lambda_{i\sigma(i)}(u,r,\Lambda)(r_{\sigma(i)}-t_{\sigma(i)})\geq \lambda_{ia}(u,r,\Lambda)(r_{a}-t_{a})$. Now, for each $i\in N$ and $a\in A$ such that $u_i(r_{\sigma(i)},\sigma(i))>u_i(r_a,a)$, $\lambda_{ia}(u,r,\Lambda)=\Lambda$. Thus, for such $i$ and $a$, $\lambda_{i\sigma(i)}(u,r,\Lambda)(r_{\sigma(i)}-t_{\sigma(i)})\geq \lambda_{ia}(u,r,\Lambda)(r_{a}-t_{a})$ if and only if $\lambda_{i\sigma(i)}(u,r)(r_{\sigma(i)}-t_{\sigma(i)})\geq \Lambda(r_{a}-t_{a})$. Since for each $i\in N$, $\lambda_{i\sigma(i)}(u,r)>0$  and for each $b\in A$, $\lambda_{i\sigma(i)}(u,r)>0$ and $r_b-t_b>0$, one can select $\Lambda>0$ satisfying the property of Lemma~\ref{Lm:boundLambda} and for which all these inequalities are also satisfied.
\end{proof}

\begin{proof}[\textit{Proof of Lemma~\ref{Lem:Perturbation-hmaxmin}}]
Consider a rebate and reshuffle at $(r,\sigma)$, i.e.,  a vector $x:=(x_a)_{a\in A}\in\R^A_{++}$ and an assignment $\mu:N\rightarrow A$ that induce allocation $(r-x,\mu)$. Fix  $\Lambda$ satisfying the property in Lemma~\ref{Lm:boundLambda} for $r$. For each $i\in N$, let $\hat u_i$ be the function $(y_a,a)\in\R\times A\mapsto \hat u_i(y_a,a):=\log\lambda_{ia}(u,r,\Lambda)+y_a$.

Let $\varepsilon>0$ be such that (i) for each  $i\in N$ and each $a\in A$ such that $u_i(r_{\sigma(i)},\sigma(i))>u_i(r_a,a)$, we have that $u_i(r_{\sigma(i)},\sigma(i))>u_i(r_a-2\varepsilon,a)$; and (ii) for each $(i,a)\in SB^u(r)$, $r_a-\varepsilon>b_i$. This $\varepsilon$ is well defined because preferences $u$ are continuous.

Fix $\delta\in(0,\varepsilon]$.

\textbf{\textit{Step 1}}: Let $(y^\delta,\gamma)\in F(N,A,\hat u,\sum_{a\in A}y^\delta_a)$ be such that $\sum_{a\in A}\exp(y^\delta_a)=\delta$.  Let $r^\delta:= (r_a-\exp(y^\delta_a))_{a\in A}$.  We claim that $(r^\delta,\gamma)\in F(N,A,u,m-\delta)$.

For simplicity, for each $i\in N$, let $\lambda_{ia}:=\lambda_{ia}(u,r)$ and $\hat\lambda_{ia}:= \lambda_{ia}(u,r,\Lambda)$. Since $\sum_{a\in A}\exp(y^\delta_a)=\delta$, and $\sum_{a\in A}r_a=m$, then $\sum_{a\in A}r^\delta_a=m-\delta$. Since $\sum_{a\in A}\exp(y^\delta_a)=\delta$, for each $a\in A$, $\exp(y^\delta_a)<\delta\leq\varepsilon$. Since  $\hat u$ is quasi-linear and $\gamma$ admits an envy-free allocation for an economy with preferences $\hat u$, $\gamma$  maximizes the summation of values for $\hat u$ \citep{Svensson-1983-Eca}. By Lemma~\ref{Lm:boundLambda}, $\gamma$ is a perfect matching in $\Fcal(r)$. Thus, $u_i(r_{\gamma(i)},\gamma(i))=u_i(r_{\sigma(i)},\sigma(i))$. Thus, for each $a\in A$ such that $u_i(r_{\sigma(i)},\sigma(i))>u_i(r_{a},a)$, we have that $u_i(r_{\gamma(i)},\gamma(i))=u_i(r_{\sigma(i)},\sigma(i))>u_i(r_{a}-\varepsilon,a)\geq u_i(r_{a}-\delta,a)$. Thus, $u_i(r_{\gamma(i)}-\exp(y^\delta_{\gamma(i)}),\gamma(i))> u_i(r_{a}-\delta,a)>u_i(r_{a}-y^\delta_{a},a)$.
Let  $A\in A$ be such that $u_i(r_{\sigma(i)},\sigma(i))=u_i(r_{a},a)$. Since both $(r_{\gamma(i)},\gamma(i))$ and $(r_a,a)$ maximize $u_i$ among the bundles in $(r,\sigma)$, we have that $\hat \lambda_{i\gamma(i)}=\lambda_{i\gamma(i)}$ and $\hat \lambda_{ia}=\lambda_{ia}$. Since $u_i(r_{\gamma(i)},\gamma(i))=u_i(r_a,a)$, $u_i(r_{\gamma(i)}-y^\delta_{\gamma(i)},\gamma(i))\geq u_i(r_{a}-y^\delta_{a},a)$ if and only if $\hat \lambda_{i\gamma(i)}\exp(y^\delta_{\gamma(i)})\geq \hat \lambda_{ia}\exp(y^\delta_a)$. This happens if and only if $\log\hat\lambda_{i\gamma(i)}+y^\delta_{\gamma(i)}\geq \log\hat\lambda_{ia}+y^\delta_a$. Now, $\log\hat\lambda_{i\gamma(i)}+y^\delta_{\gamma(i)}\geq \log\hat\lambda_{ia}+y^\delta_a$ holds because $(y^\delta,\gamma)\in F(N,A,\hat u,\sum_{a\in y_a})$. Thus, $u_i(r_{\gamma(i)}-\exp(y^\delta_{\gamma(i)}),\gamma(i))\geq u_i(r_{a}-\exp(y^\delta_{a}),a)$. Thus, for each pair $\{i,j\}\subseteq N$, $u_i(r^\delta_{\gamma(i)},\gamma(i))\geq u_i(r^\delta_{\gamma(j)},\gamma(j))$. Thus, $(r^\delta,\gamma)\in F(N,A,u,m-\delta)$.

\medskip
\textbf{\textit{Step 2}}: Let $(r^\delta,\gamma)$ be a solution to
\begin{equation}\max_{(r^\delta,\gamma)\in F(N,A,u,m-\delta)}\min_{i\in N}u_i(r^\delta_{\gamma(i)},\gamma(i)).\label{Eq:Maxmin-hl-eq1}\end{equation}
We claim that for each maximal weight perfect matching in $\Fcal^u(r)$, $(r^\delta,\mu)$ is a solution to~(\ref{Eq:Maxmin-hl-eq1}). Since $(r,\sigma)\in \Rcal(N,A,u,m)$, for each $a\in A$, $r^\delta_a<r_a$ \citep{A-D-G-1991-Eca,Velez-2016-Rent}.  Thus, for each $a\in A$, $r_a-r^\delta_1<\delta\leq\varepsilon$. By our choice of $\varepsilon$, there is no $i\in N$ and $a\in A$ such that, $r_a>b_i>r^\delta_a$. We prove that $\gamma$ is a perfect matching in $\Fcal(r)=(N,A,E)$. Suppose by contradiction that there is $i\in N$ and $a\in A$, such that $(i,a)\not\in E$. Then, $u_i(r_{\gamma(i)},\gamma(i))\neq u_i(r_{\sigma(i)},\sigma(i))$. Since $(r,\sigma)\in F(N,A,u,m)$, $u_i(r_{\gamma(i)},\gamma(i))<u_i(r_{\sigma(i)},\sigma(i))$. By our choice of $\varepsilon$, $u_i(r^\delta_{\sigma(i)},\sigma(i))>u_i(r_{\sigma(i)},\sigma(i))>u_i(r_{\gamma(i)}-2\varepsilon,{\gamma(i)})\geq u_i(r_{\gamma(i)}^\delta,{\gamma(i)})$. Since $\gamma(i)\neq \sigma(i)$, $(r^\delta,\gamma)\not\in F(N,A,u,m-\delta)$. This contradicts $(r^\delta,\gamma)$ is a solution to~(\ref{Eq:Maxmin-hl-eq1}). By Lemma~\ref{Lem:induceFinlinear},  we can further select $\Lambda>0$ satisfying the property of Lemma~\ref{Lm:boundLambda} for $r$ and such that
\[(r-r^\delta,\gamma)\in F\left(N,A,(\lambda_{ia}(u,r,\Lambda)(\cdot))_{i\in N,a\in A},\sideset{}{_{a\in A}}\sum r_a-r^\delta_a\right).\]
Thus,
\[(\log(r_a-r^\delta_a)_{a\in A},\gamma)\in F\left(N,A,\hat u,\sideset{}{_{a\in A}}\sum \log(r_a-r^\delta_a)\right),\]
and for each $\mu'$ that is a solution to
\begin{equation}
\max_{\gamma':N\rightarrow A,\ \gamma'\textrm{ a bijection}}\sum_{i\in N}\log(\lambda_{i\gamma'(i)}(u,r,\Lambda)) ,\label{Eq:opt-assignment}\end{equation}
we have that $\mu'$ is a perfect matching in $\Fcal(r)$ (this is the property of Lemma~\ref{Lm:boundLambda} for $r$). Since $\mu$ is a maximal weight perfect matching in $\Fcal^u(r)$, $\mu$ is also a solution to~(\ref{Eq:opt-assignment}). Since each optimal assignment in a quasi-linear economy can substitute the assignment in any envy-free allocation preserving no-envy \citep{Svensson-2009-ET,Gal-et-al-2016-JACM} (see \citep{Procaccia-Velez-Yu-2018-AAAI} for a short proof),
\[(\log(r_a-r^\delta_a)_{a\in A},\mu)\in F\left(N,A,\hat u,\sideset{}{_{a\in A}}\sum \log(r_a-r^\delta_a)\right).\]
Since $\sum_{a\in A}\exp(\log(r_a-r^\delta_a))=\delta$, by Step 1, $(r^\delta,\mu)\in F(N,A,u,m-\delta)$. By \citet[Lemma 3,][]{A-D-G-1991-Eca}, for each $i\in N$, $u_i(r^\delta_{\mu(i)},\mu(i))=u_i(r^\delta_{\gamma(i)},\gamma(i))$. Thus, $\min_{i\in N}u_i(r^\delta_{\mu(i)},\mu(i))=\min_{l\in L}u(r^\delta_{\gamma(i)},\gamma(i))$. Thus, $(r^\delta,\mu)$ is a solution to~(\ref{Eq:Maxmin-hl-eq1}).

\medskip
\textbf{\textit{Step 3}}: Concludes. Let $\mu$ be a maximal weight perfect matching in $\Fcal^u(r)$. Consider the function $\delta\in[0,\varepsilon]\mapsto(r^\delta,\mu)$ where  $(r^0,\mu)=(r,\mu)$ and for each $\delta\in(0,\varepsilon]$, $(r^\delta,\mu)$ is a solution to~(\ref{Eq:Maxmin-hl-eq1}). Since $\mu$ is a perfect matching in $\Fcal(r)$, $(r,\mu)\in F(N,A,u,m)$. By \citet[Lemma 3,][]{A-D-G-1991-Eca}, for each $i\in N$, $u_i(r_{\mu(i)},\mu(i))=u_i(r_{\sigma(i)},\sigma(i))$. Since $(r,\sigma)\in \Rcal(N,A,u,m)$, we have that $(r,\mu)\in \Rcal(N,A,u,m)$. Thus, for each $\delta\in[0,\varepsilon]$, $(r^\delta,\mu)\in  \Rcal(N,A,u,m-\delta)$. Thus, for each pair $0<\delta<\eta<\varepsilon$, and each $i\in N$,  $u_i(r^\eta_{\mu(i)},\mu(i))>u_i(r^\delta_{\mu(i)},\mu(i))$, and for each $a\in A$, $r^\delta_a>r^\eta_a$ \citep{A-D-G-1991-Eca,Velez-2016-Rent}.
\end{proof}

\begin{proof}[\textit{Proof of Lemma~\ref{Lem:ConversePerturbation}}]Consider a rebate and reshuffle at $(r,\sigma)$, i.e.,  a vector $x:=(x_a)_{a\in A}\in\R^A_{++}$ and an assignment $\mu:N\rightarrow A$ that induce allocation $(r-x,\mu)$. Fix  $\Lambda$ satisfying the property in Lemma~\ref{Lm:boundLambda} for $r$. For each $i\in N$, let $\hat u_i$ be the function $(y_a,a)\in\R\times A\mapsto \hat u_i(y_a,a):=\log\lambda_{ia}(u,r,\Lambda)+y_a$. Since $(r,\sigma)\in F(N,A,u,m)$, $\sigma$ is a perfect matching in $\Fcal(r)$. By Lemma~\ref{Lem:induceFinlinear} one can select $\Lambda$ such that $(\log(r-t),\sigma)\in F(N,A,\hat u,\sum_{a\in A}\log(r_a-t_a))$. Since $\hat u$ is quasi-linear, $\sigma$ maximizes the summation of the values for $\hat u$ \citep{Svensson-1983-Eca}. Thus, $\sigma$ is a maximal weight perfect matching in $\Fcal^u(r)$.
\end{proof}

\end{document}